\documentclass[11pt,a4paper,reqno]{amsart}%
\usepackage{amsthm,amsmath,amsfonts,amssymb,amsxtra,appendix,bookmark,dsfont,latexsym}

\makeatletter
\usepackage{hyperref}
\usepackage{delarray,a4,color}
\usepackage{euscript}
\usepackage[latin1]{inputenc}

\theoremstyle{plain}
\newtheorem{theorem}{Theorem}[section]
\newtheorem{lemma}[theorem]{Lemma}

\newtheorem{proposition}[theorem]{Proposition}

\theoremstyle{remark}
\newtheorem{remark}[theorem]{Remark}

\numberwithin{equation}{section}

\DeclareMathOperator{\Tr}{Tr}
\DeclareMathOperator{\tr}{Tr}

\def\geqslant{\ge}
\def\leqslant{\le}
\def\bq{\begin{eqnarray}}
\def\eq{\end{eqnarray}}
\def\bqq{\begin{eqnarray*}}
\def\eqq{\end{eqnarray*}}
\def\nn{\nonumber}

\def\eps{\varepsilon}
\def\wto{\rightharpoonup}

\newcommand{\norm}[1]{\left\lVert #1 \right\rVert}

\newcommand\1{{\ensuremath {\mathds 1} }}

\renewcommand{\epsilon}{\varepsilon}

\def\cF {\mathcal{F}}

\def\R {\mathbb{R}}

\def\C {\mathbb{C}}

\def\E {\mathcal{E}}
\def\cE {\mathcal{E}}

\def\R {\mathbb{R}}
\def\C {\mathbb{C}}

\def\gS{\mathfrak{S}}

\def\E {\mathcal{E}}
\def\cM {\mathcal{M}}

\def\gH{\mathfrak{H}}
\newcommand\ii{{\ensuremath {\infty}}}
\newcommand\pscal[1]{{\ensuremath{\left\langle #1 \right\rangle}}}

\renewcommand{\leq}{\leqslant}
\renewcommand{\geq}{\geqslant}

\newcommand{\EH}{\E_{\rm H}}
\newcommand{\eH}{e_{\rm H}}

\newcommand{\ENLS}{\cE_{\rm nls}}
\newcommand{\eNLS}{e_{\rm nls}}
\newcommand{\uNLS}{u_{\rm nls}}
\newcommand{\MNLS}{\cM_{\rm nls}}

\title[Mean-field approximation for trapped Bose gases]{The mean-field approximation and the non-linear Schr\"odinger functional for trapped Bose gases} 

\author[M. Lewin]{Mathieu LEWIN}
\address{CNRS \& Laboratoire de Math\'ematiques (UMR 8088), Universit\'e de Cergy-Pontoise, F-95000 Cergy-Pontoise, France.}
\email{mathieu.lewin@math.cnrs.fr}

\author[P.T. Nam]{Phan Th\`anh NAM}
\address{IST Austria, Am Campus 1, 3400 Klosterneuburg, Austria} 
\email{pnam@ist.ac.at}

\author[N. Rougerie]{Nicolas ROUGERIE}
\address{CNRS \& Universit\'e Grenoble Alpes,  LPMMC (UMR 5493), B.P. 166, F-38 042 Grenoble, France}
\email{nicolas.rougerie@grenoble.cnrs.fr}

\date{2015. Final version to appear in the \emph{Transactions of the American Mathematical Society}.\\ $\copyright$2015 by the authors.}

\begin{document}

\begin{abstract}
We study the ground state of a trapped Bose gas, starting from the full many-body Schr\"odinger Hamiltonian, and derive the nonlinear Schr\"odinger energy functional in the limit of large particle number, when the interaction potential converges slowly to a Dirac delta function.  Our method is based on quantitative estimates on the discrepancy between the full many-body energy and its mean-field approximation using Hartree states. These are proved using finite dimensional localization and a quantitative version of the quantum de Finetti theorem. Our approach covers the case of attractive interactions in the regime of stability. In particular, our main new result is a derivation of the 2D attractive nonlinear Schr\"odinger ground state.
\end{abstract}

\maketitle

\setcounter{tocdepth}{2}
\tableofcontents

\section{Introduction}\label{sec:intro}

The impressive progress of cold atom physics~\cite{PetSmi-01,PitStr-03} during the last two decades gave a new impetus to the theory of many-bosons systems, in particular the study of Bose-Einstein condensates (BECs). In such a state of matter, many interacting particles (bosons) occupy the same quantum state and may thus be collectively modeled using a one-body nonlinear Schr\"odinger description. 

In this framework, the energy functional of a Bose-Einstein condensate of non-relativistic particles in\footnote{We are interested in $d=1,2,3$ but our method also applies to $d\geq 4$.} $\R ^d$, interacting via a pair potential $w$, reads as (in appropriate units)
\begin{equation}\label{eq:intro Hartree}
\EH [u] := \int_{\R ^d } \left| \left( \nabla + i A \right) u \right| ^2 + V  |u | ^2 + \frac{1}{2} |u| ^2 (w \ast |u| ^2 ) .
\end{equation}
This so-called Hartree functional includes a trapping potential $V:\R ^d \to \R$ confining the particles in a bounded region of space ($V(x)\to \infty$ when $|x|\to \infty$), a feature ubiquitous in experiments with cold atoms. The vector potential $A:\R ^d \to \R ^d$ can model a (possibly artificial~\cite{DalGerJuzOhb-11}) magnetic field or the Coriolis force due to the rotation of the atoms~\cite{Aftalion-06,Cooper-08,Fetter-09}, in which case $V$ has to incorporate a contribution due to the centrifugal force. The ground state of the system is obtained by minimizing~\eqref{eq:intro Hartree} under the mass constraint
\begin{equation}\label{eq:intro mass}
\int_{\R ^d } |u| ^2 = 1. 
\end{equation}
Samples in which Bose-Einstein condensation is achieved are typically very dilute, and it is therefore relevant to think of the case where the range of the interaction potential $w$ is much smaller than the size of the system. This situation can be modeled by using contact interactions, which amounts to formally set $w= a \delta_0$ proportional to the delta function at the origin. One then obtains the nonlinear Schr\"odinger (or Gross-Pitaevskii) functional 
\begin{equation}\label{eq:intro nls}
\ENLS [u] := \int_{\R ^d } \left| \left( \nabla + i A \right) u \right| ^2 + V  |u | ^2 + \frac{a}{2} |u| ^4.
\end{equation}
Here $a$ measures the strength of interparticle interactions. An appealing aspect of cold atoms experiments is the possibility to tune the value and even the sign of $a$, going from repulsive ($a>0$) to attractive interactions ($a<0$)\footnote{Respectively from defocusing to focusing interactions in the quantum optics vocabulary.}. In the latter case the system may collapse, i.e. there might not exist a ground state for~\eqref{eq:intro nls}. It is in fact the case if $d=2$ and $a\leq -a^*$ for some critical value $a ^* >0$, and also if $d\geq 3$ and $a<0$. 

An important question is that of the relation between the macroscopic, nonlinear description using functionals such as~\eqref{eq:intro Hartree} and~\eqref{eq:intro nls}, and the underlying, linear, physics based on the many-body Schr\"odinger equation.  
In the latter model the system is described via a many-body Hamiltonian of the form
\begin{equation}\label{eq:start hamil}
H_N = \sum_{j=1} ^N  \left( - \left( \nabla_{x_j} +i A (x_j) \right) ^2 + V(x_j) \right) + \frac{1}{N-1} \sum_{1\leq i<j \leq N}N ^{d\beta} w( N ^{\beta} (x_i-x_j))
\end{equation}
acting on $\gH ^N := \bigotimes_s ^N\gH$, the symmetric tensor product of $N$ copies of the one-body Hilbert space $\gH = L ^2 (\R ^d )$ describing quantum particles living in $d$ dimensions. The symmetry restriction is necessary for bosonic particles that correspond to square integrable complex wave functions $\Psi\in L ^2 (\R ^{dN}) \simeq \bigotimes ^N L ^2 (\R ^d)$ of $N$ variables satisfying
\begin{equation}\label{eq:Bose symmetry}
\Psi (x_1,\ldots,x_N) = \Psi(x_{\sigma(1)}, \ldots, x_{\sigma(N)}) 
\end{equation}
for any permutation $\sigma$. The ground state energy associated with~\eqref{eq:start hamil} is simply the smallest eigenvalue of the operator, and a ground state is any associated eigenfunction. Roughly speaking, one should expect that when $N\to \infty$, a ground state of~\eqref{eq:start hamil} factorizes:
\begin{equation}\label{eq:formal BEC}
\Psi (x_1,\ldots,x_N)  \approx \prod_{j=1} ^N u (x_j),  
\end{equation}
with $u\in L ^2 (\R ^d )$ a ground state of the nonlinear one-body theory. Of course~\eqref{eq:formal BEC} must be taken with care: determining in which precise sense it holds is a subtle task that has motivated much research.

In~\eqref{eq:start hamil} we have scaled the interaction term in such a way that we may expect a well-defined theory in the limit $N\to \infty$. The prefactor $(N-1) ^{-1}$ ensures that the interaction energy stays of order $N$, and the fixed parameter $0\leq \beta \leq 1$ allows to model different interacting scenarii by changing the $N$-dependence of the potential range. For illustration, let us recall how the limit problem depends on $\beta$ in three space dimensions:\begin{itemize}
\item If $\beta = 0$ we are in the mean-field (MF) regime. The range of the interaction potential is fixed but its intensity goes to zero proportionally to $N ^{-1}$. In this case the limit problem is~\eqref{eq:intro Hartree}. The relation between~\eqref{eq:intro Hartree} and~\eqref{eq:start hamil} has been clarified under special assumptions on $w$ in~\cite{BenLie-83,LieYau-87,Seiringer-11,GreSei-13,SeiYngZag-12} and in full generality in~\cite{LewNamRou-14}, taking inspiration from the earlier works~\cite{FanSpoVer-80,RagWer-89,Werner-92}.
\item If $ 0 < \beta < 1$, the scaled interaction potential converges in the sense of measures to a delta function
\begin{equation}\label{eq:interaction delta}
w_N := N ^{d\beta} w ( N ^{\beta} \cdot) \wto \left( \int_{\R ^d} w \right)\delta_0  
\end{equation}
and the limit problem becomes~\eqref{eq:intro nls} with $a = \int_{\R ^d} w$. We will refer to this case as the nonlinear Schr\"odinger (NLS) limit. It has not been considered much in the literature. However, when $w\ge 0$, the techniques invented to tackle the harder case $\beta = 1$ a fortiori apply in this case.
\item If $\beta = 1$, the interaction potential of course still converges to a delta function but, in dimension $d=3$, the limit functional is now~\eqref{eq:intro nls} with $a = 8\pi \times$(scattering length of $w$). This is due to the fact that the ground state of~\eqref{eq:start hamil} includes a non trivial correction to the ansatz~\eqref{eq:formal BEC}, in the form of a short scale correlation structure. The relationship between~\eqref{eq:start hamil} and~\eqref{eq:intro nls} in this so-called Gross-Pitaevskii (GP) limit has been established by Lieb, Seiringer and Yngvason in a seminal series of papers (we refer for example to~\cite{LieYng-98,LieSeiYng-00,LieSeiYng-05,LieSei-06}, or~\cite{LieSeiSolYng-05} for a review, see also~\cite{BruCorPicYng-08}), assuming the interactions to be purely repulsive, $w\geq 0$.
\end{itemize}
The corresponding evolution problems have also attracted a lot of attention recently. Among many references, one may consult~\cite{BarGolMau-00,AmmNie-08,FroKnoSch-09,RodSch-09,Pickl-11} for the MF limit,~\cite{AdaBarGolTet-04,ErdSchYau-07,Pickl-10} for the NLS limit and~\cite{ErdSchYau-09,Pickl-10b} for the GP limit. For the derivation of the 1D focusing NLS equation, we refer to~\cite{CheHol-13}.

\medskip

The purpose of this paper is to present a new method allowing to rigorously establish that the minimization of~\eqref{eq:intro nls} correctly describes the ground state of~\eqref{eq:start hamil} in the limit $N\to \infty$. We are limited to rather small values of $\beta$, i.e. to potentials converging slowly (with polynomial rate) to delta interactions, so we always obtain~\eqref{eq:intro nls} with 
$$ a = \int_{\R ^d} w.$$
Our method consists in establishing quantitative bounds on the difference between the $N$-body energy per particle and an $N$-dependent Hartree energy obtained by taking 
$$w_N = N ^{d\beta} w (N ^{\beta} .)$$
as an interaction potential in~\eqref{eq:intro Hartree}. The NLS functional is obtained in a second step by passing to the limit in the Hartree functional. We thus treat separately the issues of the mean-field approximation  (the validity of the ansatz~\eqref{eq:formal BEC}) and that of the short range of the interactions (the nature of the one-body state $u$ in the ansatz~\eqref{eq:formal BEC}). 

To deal with the first (main) issue we elaborate on the method we presented in~\cite{LewNamRou-14}, which is based on the quantum de Finetti theorem~\cite{Stormer-69,HudMoo-75}. However, in~\cite{LewNamRou-14} we relied heavily on compactness arguments and thus did not obtain any error bound, which is required to deal with a $N$-dependent interaction potential as we do here. Our strategy in this paper is to use the localization method in Fock space described in~\cite{Lewin-11} to reduce the problem to a finite dimensional setting in which we may employ a quantitative version of the quantum de Finetti theorem due to Christandl, K\"onig, Mitchison and Renner~\cite{ChrKonMitRen-07}.

The intuition behind this procedure is as follows. Under the assumption that the system is trapped, namely $V(x) \to \infty$ when $|x|\to \infty$, the one-body operator
\begin{equation}\label{eq:intro one body}
 H_1 := -\left( \nabla + i A (x) \right) ^2 + V (x)
\end{equation}
of the Hamiltonian~\eqref{eq:start hamil} has a discrete spectrum made of a diverging sequence of eigenvalues. In the ground state of the $N$-body Hamiltonian, the number of particles living on one-body states with high energy will clearly be small. The idea is thus to restrict our attention to the subsystem consisting of all the particles which have a one-particle energy below a given cut-off $L$, that is allowed to tend to infinity slowly. 

The part of the many-body ground state living in the low energy space will be dealt with using a quantitative version of the finite dimensional quantum de Finetti theorem~\cite{FanVan-06,ChrKonMitRen-07,Chiribella-11,Harrow-13,LewNamRou-13b}. This result roughly says that the reduced density matrices 
of any $N$-body state can be approximated by that of a convex combination of product states $|u ^{\otimes N} \rangle \langle u^{\otimes N}|$. The energy being a linear functional of the $2$-body density matrix, it is then easy to obtain the Hartree energy as a lower bound. The approximation error due to this procedure is proportional to the dimension of the low-lying energy space and inversely proportional to the number of particles. A crucial step therefore consists in optimizing over the energy cut-off $L$ (which governs the dimension of the low energy subspace) to minimize the error due to the use of the de Finetti theorem.

A related approach was used by Lieb and Seiringer in~\cite{LieSei-06} who dealt with the Gross-Pitaevskii limit of the same model. After having replaced the interaction by a smeared one involving the scattering length (a step which is not considered here), they also localized the particles to the lower energy states. Then they used coherent states in Fock space for the particles with the lower energies, in order to recover the mean-field functional. An advantage of our approach is that it works ``locally in Fock space'', that is, in any $n$-particle subspace independently of the others and thus avoids some difficulties related to the control of high values of the particle number.

Our method seems less sensitive to the type of interaction potential considered than previous approaches. In particular we do not need $w\geq 0$ to establish quantitative estimates on the difference between the $N$-body energy and the $N$-dependent Hartree energy. However, we need a sufficiently stable system to obtain the NLS energy when passing to the limit $N\to \infty$. In fact, the $N$-body system must be \emph{stable of the second kind}, namely the lowest energy per particle must be uniformly bounded. In the following, we shall only use  stability conditions for the $N$-dependent Hartree functional. Since the Hartree energy is always an upper bound to the $N$-body energy, the Hartree stability is obviously {\em necessary} for the $N$-body stability. That the Hartree stability is {\em sufficient} for the $N$-body stability is not obvious, and it comes from our quantitative estimates on the energy difference 
between the $N$-body and Hartree models. To be precise, for $\beta$ small enough (depending on $d$ and the trapping potential $V$), we obtain the NLS energy (and NLS ground state) from the $N$-body model when $N\to \infty$ in the following cases
\begin{enumerate}
\item $d=3$ and the interaction potential is {stable} in the sense that 
\begin{equation}
\iint_{\R^3\times \R ^3}|u(x)| ^2 |u(y)| ^2 w(x-y)\,dx\,dy\geq 0,\qquad \forall u\in L ^2 (\R ^d).
\label{eq:repul 3D}
\end{equation}
Assumption~\eqref{eq:repul 3D} is {\em necessary} in dimensions $d\geq3$. In fact, \eqref{eq:repul 3D} is necessary (and sufficient if $w$ is regular enough) for the stability of the second kind of {\em classical particles} interacting via the potential $w$, see Subsection~\ref{sec:stability} below. 

\item $d=2$ and the potential is stable in the sense that
$$\inf_{u\in H^1(\R^2)}\left(\frac{\displaystyle\iint_{\R^2\times \R ^2}|u(x)|^2|u(y)|^2w(x-y)\,dx\,dy}{2\norm{u}_{L^2(\R^2)}^2\norm{\nabla u}_{L^2(\R^2)}^2}\right)>-1.$$
Except for the case of equality, this assumption is again also {\em necessary}. 
\item $d=1$ under no specific assumption on the interaction potential.
\end{enumerate}

See Section \ref{sec:stability} for a detailed discussion on these assumptions. Of these three cases, the second is the one that presents the main novelty of the paper: in 2D, we present the first derivation of the attractive NLS ground state in the regime of stability. In 3D on the other hand our conditions on $\beta$ are much more stringent than those in~\cite{LieSeiYng-00,LieSei-06} where the GP limit is covered in the case $w\geq 0$. The result is thus not new but our method of proof may still be of interest, in particular because it provides error bounds, and treats on the same footing the case with and without vector potential $A$. In 1D, because of the Sobolev embedding, one may directly use contact interactions at the level of the many-body Hamiltonian~\eqref{eq:start hamil} as in~\cite{LieLin-63,SeiYngZag-12}, so the procedure of scaling a regular interaction potential is less relevant. However, we provide a derivation of the attractive 1D ground state with what seems to be an 
unprecedented precision on the error bound. 

Our main theorems are stated in the next section, and their proofs occupy the rest of the paper. 

\medskip

\noindent \textbf{Acknowledgement:} The authors acknowledge financial support from the European Research Council (FP7/2007-2013 Grant Agreement MNIQS 258023) and the ANR (Mathostaq project, ANR-13-JS01-0005-01). PTN and NR have benefited from the hospitality of the \emph{Institute for Mathematical Science} of the National University of Singapore.

\section{Main results}\label{sec:main results}


We do not aim at optimal assumptions on the potentials $V$ and $w$. We will assume that $w$ is  integrable, symmetric and uniformly bounded:
\begin{equation}\label{eq:asum w}
w\in L^1(\R^d,\R), \quad w(x) = w(-x) \quad {\rm and}\quad |w(x)| \leq C .
\end{equation}
We shall denote by
\begin{equation*}
w_N (x) := N ^{d\beta} w (N ^{\beta}x)
\end{equation*}
the scaled interaction potential, and use the same notation $w_N$ for the operator acting on the two-particle space $\gH ^2\simeq L ^2 (\R ^{2d})$ as the multiplication by $(x_1,x_2)\mapsto w_N (x_1-x_2)$. 

The fact that we consider a trapped system, as appropriate for experiments with cold atomic gases, is materialized by the following assumption on the one-body potential: 
\begin{equation}\label{eq:asum V}
V\in L^1_{\rm loc}(\R^d, \R), \quad V(x) \ge c |x| ^s-C 
\end{equation}
for some exponent $s>0$ and some constants $c>0$, $C\geq0$. 
Our estimates will essentially involve the exponent $s$. 

As for the vector potential $A$, that can model a  magnetic field or Coriolis forces acting on rotating particles, it is sufficient to assume 
\begin{equation}\label{eq:asum A}
A \in L ^2_{\rm loc} (\R ^d, \R ^d).
\end{equation}
A particularly relevant example is given by $A(x) = \Omega (-x_2,x_1,0)$ for $d=3$ and $A(x)=\Omega(-x_2,x_1)$ for $d=2$, corresponding to Coriolis forces due to a rotation at speed $\Omega$ around the $x_3$-axis, or a constant magnetic field of strength $\Omega$ pointing in this direction. 

\subsection{Error bounds for the mean-field approximation}

Our first main task is to provide quantitative bounds on the discrepancy between the ground state energy per particle corresponding to~\eqref{eq:start hamil}
\begin{equation}\label{eq:gse}
E(N) := \inf \sigma_{\gH ^N} H_N  
\end{equation}
and the nonlinear energy 
\begin{equation}\label{eq:hartree e}
\eH := \inf_{\norm{u}_{L ^2 (\R ^d)} = 1} \EH [u] 
\end{equation}
given by the minimization of the Hartree functional
\begin{multline}\label{eq:hartree f}
\EH [u] := \int_{\R ^d} \left( \left|\left(\nabla + i A(x)  \right) u(x) \right| ^2 + V(x) |u(x)| ^2 \right) dx\\ + \frac{1}{2} \iint_{\R ^d \times \R ^d} |u(x)| ^2 w_N (x-y) |u(y)| ^2 dx dy.  
\end{multline}
Note that when $\beta >0$ both $\EH$ and $\eH$ depend on $N$ and that $|\eH|\leq CN^{d\beta}$. 

It will be convenient to introduce a slightly modified Hartree energy with the interaction $w_N$ replaced by $w_N-\epsilon|w_N|$ for some $\epsilon>0$ (that will later on be taken small enough):
\begin{equation}\label{eq:hartree f modified}
\EH^\epsilon [u] := \EH[u]- \frac{\epsilon}{2} \iint_{\R ^d \times \R ^d} |u(x)| ^2 |w_N(x-y)| |u(y)| ^2 dx dy.  
\end{equation}
We denote by
\begin{equation}\label{eq:hartree e modified}
\eH^\epsilon := \inf_{\norm{u}_{L ^2 (\R ^d)} = 1} \EH^\epsilon [u] 
\end{equation}
the corresponding ground state energy, which satisfies $|\eH^\epsilon|\leq C(1+\epsilon)N^{d\beta}$. 

When $\beta > 0$ we also consider the NLS energy functional
\begin{equation}\label{eq:nls f}
\ENLS [u] = \int_{\R ^d} \left( \left|\left(\nabla + i A(x)  \right) u(x) \right| ^2 + V(x) |u(x)| ^2 \right)dx + \frac{a}{2}  \int_{\R ^d} |u(x)| ^4 dx  
\end{equation}
with ground state energy 
\begin{equation}\label{eq:nls e}
\eNLS = \inf_{\norm{u}_{L ^2 (\R ^d)} = 1} \ENLS [u] 
\end{equation}
that arises as the limit $N\to \infty$ of~\eqref{eq:hartree f}. As announced, $a$ will always be defined as
\begin{equation}\label{eq:int w}
a:=\int_{\R ^d} w . 
\end{equation}
The ground state energy $\eNLS$ is finite in dimensions $d\geq2$ only under appropriate assumptions on $a$ (e.g. $a\geq 0$).

For $\beta = 0$ it is well-known (see~\cite{LewNamRou-14} and references therein) that 
\begin{equation}\label{eq:Hartree lim}
\lim_{N\to \infty} \frac{E(N)}{N} = \eH. 
\end{equation}
For the purpose of this paper we need to provide explicit estimates in the case of a confined system and for $\beta \ge 0$, which is the content of the following

\begin{theorem}[\textbf{Error bounds for the mean-field approximation}]\label{thm:error Hartree}\mbox{}\\
We assume that~\eqref{eq:asum w},~\eqref{eq:asum V} and~\eqref{eq:asum A} hold true. 

\smallskip

\noindent$\bullet$ If $d=1$ and $\beta>0$, then we have
\begin{equation}\label{eq:energy estimate 1D}
\eH \geq \frac{E(N)}{N} \geq \eH - CN^{-\tfrac{1}{4+2/s}}.
\end{equation}

\smallskip

\noindent$\bullet$ If $d=1$ and $\beta=0$, or if $d\geq2$ and
\begin{equation}\label{eq:beta 0}
0 \leq \beta < \frac{1}{d(1+d/s+d/2)},
\end{equation}
then we have
\begin{equation}\label{eq:energy estimate}
\eH \geq \frac{E(N)}{N} \geq \eH^\epsilon - C\frac{\epsilon^{-1-d/2-d/s}}{N ^{1-d\beta(1+d/2+d/s)}},
\end{equation}
for all $0<\epsilon\leq1$.
\end{theorem}

Without more information on the potential $w$ and when $\beta>0$, it is not obvious that the problem $\eH^\epsilon$ is actually close to $\eH$ and we will only be able to go further under some additional stability assumptions on $w$. We can however make the following 

\begin{remark}[Estimates in the mean-field limit]\label{rem:MF estimates}\mbox{}\\
Using the simple estimate
$$\eH^\epsilon\geq \eH-CN^{d\beta}\epsilon$$ 
and optimizing with respect to $\epsilon$, we can immediately deduce from~\eqref{eq:energy estimate} a bound which is valid for a smaller range of $\beta$ but does not involve $\eH^\epsilon$ anymore. Namely, if 
\begin{equation}\label{eq:beta 0 worse}
0 \leq \beta < \frac{1}{d(2+d/s+d/2)},
\end{equation}
then we have
\begin{equation}\label{eq:worse energy estimate}
\eH \geq \frac{E(N)}{N} \geq \eH - CN ^{d\beta-\tfrac{1}{2+d/s+d/2}}.
\end{equation}
If $\beta=0$ then $\eH$ is independent of $N$ and we have thus the uniform bound
\begin{equation}
\eH\geq \frac{E(N)}{N}\geq \eH-CN ^{-\tfrac{1}{2+d/2+d/s}}
\label{eq:estim Hartree}
\end{equation}
which gives a bound on the rate of convergence in~\eqref{eq:Hartree lim}.\hfill $\triangle$
\end{remark}

With more information on $w$, it is clear that~\eqref{eq:estim Hartree} is not optimal. Indeed, it has been proved in~\cite{LewNamSerSol-13,Seiringer-11,GreSei-13} that when the Hartree minimization problem $\eH$ has a unique non-degenerate minimizer, then the next order is given by Bogoliubov's theory and it is of order $1/N$. This is for instance the case if\footnote{$\widehat{w}$ denotes the Fourier transform of $w$.} $\widehat{w}>0$ and the vector potential $A$ is small enough. Without more assumptions on $w$, we are not aware of any existing quantitative error bound. For instance, in a gas rotating sufficiently fast for vortices to be nucleated (see~\cite{CorPinRouYng-11a,CorPinRouYng-12,CorRinYng-07,Sei-02} and references therein), uniqueness is known to fail and our bound~\eqref{eq:estim Hartree} seems to be the best so far.

For illustration, we compute the rate of convergence we obtain in a few physically interesting situations (see Table~\ref{tab:MF}), still in the Hartree case $\beta=0$. We consider space dimensions $d=1,2,3$, and compare the harmonic oscillator case $s=2$ with the case of particles in a box where we set formally $s=\infty$. In fact our method also applies to the case of particles confined to a bounded domain and we obtain the rates for $s=\infty$ in this case.
\begin{table}[htbp]
\begin{tabular}{|c|c|c|c|}
\hline
 & $d=3$ & $d=2$ & $d=1$ \\
\hline
$\phantom{\big|}s=2$  & $N ^{-1/5}$ & $N ^{-1/4}$ & $N ^{-1/3}$\\
\hline
$\phantom{\big|}s=\infty$ & $N ^{-2/7}$  & $N ^{-1/3}$ & $N ^{-2/5}$\\
\hline
\end{tabular}
\vspace{0.2cm}
\caption{Rates of convergence to Hartree's energy ($\beta=0$).}\label{tab:MF} 
\end{table}

\vspace{-0.7cm}

\subsection{Error bounds for the Non Linear Schr\"odinger model}

If $\beta>0$, then $\eH$ and $\eH^\epsilon$ still depend on $N$ in our bounds~\eqref{eq:energy estimate 1D} and~\eqref{eq:energy estimate}. Our next task is to relate these energies to the NLS ground state energy $\eNLS$. In general, $\eH$ will not converge to $\eNLS$. For this to be true, some stability properties of the interaction potential $w$ are needed, further discussed in the following. 

\subsubsection{Stability properties of $w$}\label{sec:stability}

An even potential $W$ is called \emph{classically stable} ~\cite{Ruelle} when
\begin{equation}
\sum_{1\leq i<j\leq N} W(x_i-x_j)\geq -CN,\qquad \forall x_1,...,x_N\in\R^d,\ \forall N\geq2.
\label{eq:repulsive-Ruelle}
\end{equation}
By integrating against a factorized measure $\rho(x_1)\cdots \rho(x_N)$ and letting $N\to\ii$, we see that~\eqref{eq:repulsive-Ruelle} implies 
\begin{equation}
\iint_{\R^d\times \R ^d} \rho(x)\rho(y)W(x-y)\,dx\,dy\geq 0,\qquad \forall \rho\geq 0.
\label{eq:repulsive}
\end{equation}
Conversely, when $W$ is bounded,~\eqref{eq:repulsive} implies~\eqref{eq:repulsive-Ruelle} as is seen by taking $\rho=\sum_{i=1}^N\delta_{x_i}$. By dilating $\rho$ one can see that \eqref{eq:repulsive} implies $\int_{\R^d} W\geq0$.

The relevance of classical stability to our quantum problem depends on the dimension:

\noindent$\bullet$ When $d\ge 3$, the nonlinear term of the NLS functional is super-critical with respect to the kinetic energy and we shall take \eqref{eq:repulsive} to be our assumption. Note that it is really necessary to ensure that $e_{\rm H}$ does not converge to $-\infty$, which can be seen immediately by taking a trial state $u_N (x)=N^{d\beta/2}u(N^\beta x)$. In other words, in 3D, stability of the quantum problem requires classical stability.

\noindent$\bullet$ In dimension $d=2$, the nonlinear term of the NLS functional is critical with respect to the kinetic energy. Classical stability is then not the optimal concept and Condition~\eqref{eq:repulsive} can be relaxed a bit with the help of the kinetic energy. We say that an even potential $W$ is \emph{Hartree-stable} when 
\begin{equation}
\norm{u}_{L^2}^2\norm{\nabla u}_{L^2}^2+\frac12\iint_{\R^2\times \R^2}|u(x)|^2|u(y)|^2W(x-y)\,dx\,dy\geq0
\label{eq:Hartree stable}
\end{equation}
for all $u\in H^1(\R^2)$. Clearly, a classically-stable potential is also Hartree-stable. Replacing $u$ by $\lambda u(\lambda x)$ and taking the limit $\lambda\to 0$
, we see that such a potential $W$ must satisfy
$$\norm{u}_{L^2}^2\norm{\nabla u}_{L^2}^2+\frac12\left(\int_{\R^2}W\right)\int_{\R^2}|u(x)|^4\,dx\geq0,\qquad\forall u\in H^1(\R^2).$$
This exactly means that 
$$\int_{\R^2}W(x)\,dx\geq-a^*$$
where $a^*$ is the critical interaction strength~\cite{Weinstein-83,Zhang-00,GuoSei-13,Maeda-10} for existence of a ground state for the NLS functional, that is, $a ^* := \norm{Q}_{L ^2 (\R ^2)} ^2 $, where $Q\in H ^1 (\R ^2)$ is the unique (up to translations) positive radial solution of 
\begin{equation}\label{eq:2D model prob}
-\Delta Q + Q - Q ^3 = 0.
\end{equation}
On the other hand, using 
the Cauchy-Schwarz and Young's inequalities\footnote{with $W^-:=-\min(0,W)$ the negative part of $W$} we have
$$\iint_{\R^2\times\R^2}|u(x)|^2|u(y)|^2W(x-y)\,dx\,dy\geq - \| u^2\|_{L^2} \| |u|^2*W^-\|_{L^2} \ge  -\left(\int_{\R^2}W^-\right)\int_{\R^2}|u(x)|^4\,dx,$$
and we see that $\int_{\R^2}W^-\leq a^*$ implies~\eqref{eq:Hartree stable}. In the following we shall actually need a slightly stronger notion of Hartree-stability, obtained by requiring~\eqref{eq:Hartree stable} with a strict inequality,
\begin{equation}
\inf_{u\in H^1(\R^2)}\left(\frac{\displaystyle\iint_{\R^2\times \R^2}|u(x)|^2|u(y)|^2W(x-y)\,dx\,dy}{2\norm{u}_{L^2(\R^2)}^2\norm{\nabla u}_{L^2(\R^2)}^2}\right)>-1
\label{eq:Hartree stable strict}
\end{equation}
which plays the same role as the assumption $\int_{\R^2}W>-a^*$ in the NLS case.  

\noindent$\bullet$ In dimension $d=1$, the NLS nonlinearity is subcritical and the quantum system is always stable thanks to the kinetic energy, whence the absence of specific assumptions in this case.

\medskip

In order to simplify our presentation, from now on we use the word ``stable'' for a potential $W$ that satisfies~\eqref{eq:Hartree stable strict} in dimension $d=2$ and~\eqref{eq:repulsive} in dimension $d\geq3$. The importance of these concepts is illustrated in the following.

\begin{proposition}[\bf Convergence of $\eH$ towards $\eNLS$]\label{thm:CV Hartree}\ \\
Let $\beta>0$. We assume that~\eqref{eq:asum w},~\eqref{eq:asum V} and~\eqref{eq:asum A} hold true, and further suppose that $|x|w (x) \in L ^1 (\R ^d)$. 

\smallskip

\noindent $\bullet$ If $d=1$, or $d=2,3$ with $w$ stable, then we have
\begin{equation}\label{eq:energy Hartree NLS}
|\eH-\eNLS|\leq CN^{-\beta}.
\end{equation}
Furthermore, minimizers for $\eH$ converge in the limit $N\to\ii$ to a minimizer for $\eNLS$, after extraction of a subsequence.

\smallskip

\noindent$\bullet$ If $d=2$ and 
\begin{equation}
\inf_{u\in H^1(\R^2)}\left(\frac{\displaystyle\iint_{\R^2\times\R^2}|u(x)|^2|u(y)|^2w(x-y)\,dx\,dy}{2\norm{u}_{L^2(\R^2)}^2\norm{\nabla u}_{L^2(\R^2)}^2}\right)<-1,
\label{eq:Hartree unstable}
\end{equation}
or if $d\geq3$ and $w$ is not classically stable (in the sense of~\eqref{eq:repulsive}), then 
\begin{equation}\label{eq:Hartree unstable res}
\lim_{N\to\ii}\eH = -\ii. 
\end{equation}
\end{proposition}

The decay condition $|x|w(x)\in L^1$ is a technical assumption which allows us to obtain~\eqref{eq:energy Hartree NLS} by controling the error induced by replacing the interaction potential $w_N$ in the Hartree model with the delta-potential in the NLS model. Without this assumption we still have $\eH \to \eNLS$ but with no estimate on the convergence rate. Of course~\eqref{eq:Hartree unstable res} immediately implies 
\begin{equation}\label{eq:Hartree unstable 2}
\lim_{N\to\ii}\frac{E(N)}{N} = -\ii
\end{equation}
since $E(N) \leq N\eH$. In the above result, the stability condition on $w$ is optimal in dimension $d=3$. If $w$ is classically unstable we have~\eqref{eq:Hartree unstable 2} but it could be that $\int_{\R^3}w\geq0$, and then $\eNLS$ is finite and therefore cannot be related to the limit of the $N$-body problem. In dimension $d=2$ we are only missing the case of equality in~\eqref{eq:Hartree stable strict} and the stability condition is therefore nearly optimal.

\subsubsection{Convergence of the many-body problem}

With Theorem~\ref{thm:error Hartree} and Proposition~\ref{thm:CV Hartree} at hand, it is an easy task to deduce error bounds between $E(N)/N$ and the NLS minimum energy.

\begin{theorem}[\textbf{Derivation of the NLS ground state energy}]\label{thm:deriv nls}\mbox{}\\
Let $\beta>0$. We assume that~\eqref{eq:asum w},~\eqref{eq:asum V} and~\eqref{eq:asum A} hold true, and further suppose that
$|x|w (x) \in L ^1 (\R ^d)$.

\smallskip

\noindent$\bullet$ If $d=1$, then we have
\begin{equation}\label{eq:energy estimate 1D NLS}
CN^{-\beta}+\eNLS\geq \frac{E(N)}{N}\geq \eNLS- CN^{-\beta}-CN^{-\tfrac{1}{4+2/s}}
\end{equation}
for all $\beta>0$.
\smallskip

\noindent$\bullet$ If $d=2$ and $w$ is stable (in the sense of~\eqref{eq:Hartree stable strict}), or if $d=3$ and $w-\eta|w|$ is stable (in the sense of~\eqref{eq:repulsive}) for some $0<\eta<1$, then 
\begin{equation}\label{eq:energy estimate NLS}
\eNLS+CN^{-\beta} \geq \frac{E(N)}{N} \geq \eNLS - CN^{-\beta} - CN^{-\tfrac{1-d\beta(1+d/2+d/s)}{2+d/s+d/2}},
\end{equation}
provided that $\beta$ satisfies~\eqref{eq:beta 0}.

\smallskip

\noindent$\bullet$ Finally, if $d=3$ and $w$ is stable (in the sense of ~\eqref{eq:repulsive}), then we have 
\begin{equation}\label{eq:energy estimate NLS worse}
\eNLS+CN^{-\beta} \geq \frac{E(N)}{N} \geq \eNLS - CN^{-\beta} - CN^{d\beta-\tfrac{1}{2+d/s+d/2}}
\end{equation}
provided that $\beta$ satisfies~\eqref{eq:beta 0 worse}.
\end{theorem}

We remark that when $d=2$, since the stability assumption ~\eqref{eq:Hartree stable strict} is strict, the stability for $w$ implies the stability for $w-\eta|w|$ for some $0<\eta<1$. However, the same does not hold when $d=3$ and the error estimate \eqref{eq:energy estimate NLS} is really better than \eqref{eq:energy estimate NLS worse} (and similarly for the range of $\beta$).

The proofs of Theorems~\ref{thm:error Hartree} and~\ref{thm:deriv nls} also imply some estimates on the ground states themselves. In the absence of any assumption on the behavior of the Hartree functional we cannot convert them into convergence rates for the states. In the NLS limit we nevertheless obtain convergence of states and Bose-Einstein condensation: 

\begin{theorem}[\textbf{Convergence of states in the NLS limit}]\label{thm:nls state}\mbox{}\\
We use the same assumptions as in Theorem~\ref{thm:deriv nls}. Denote $\Psi_N$ a ground state of the many-body Hamiltonian~\eqref{eq:start hamil} and 
$$ \gamma_N ^{(n)}:= \tr_{n+1\to N}  |\Psi_N\rangle \langle \Psi_N | $$
its $n$-body reduced density matrix. We have, modulo a subsequence,  
\begin{equation}\label{eq:state convergence}
\lim_{N\to \infty}\gamma_N ^{(n)} = \int_{u\in \MNLS} d\mu (u) |u ^{\otimes n} \rangle \langle u ^{\otimes n}| 
\end{equation}
strongly in the trace-class for any $n\ge1$, with $\mu$ a Borel probability measure supported on 
\begin{equation}\label{eq:nls set}
\MNLS = \left\{ u \in L ^2 (\R ^d),\ \norm{u}_{L ^2} = 1,\ \ENLS [u] = \eNLS \right\}. 
\end{equation}
In particular, when the NLS ground state $\uNLS$ is unique (modulo a constant phase), we have convergence for the whole sequence
\begin{equation}\label{eq:BEC}
\lim_{N\to \infty} \gamma_N ^{(n)} =|\uNLS ^{\otimes n} \rangle \langle \uNLS ^{\otimes n}|, \mbox{ strongly in trace-class norm.}
\end{equation}
\end{theorem}

Uniqueness of $\uNLS$ is ensured when $\int w$ is either $\geq 0$ or small enough in absolute value in dimensions $d=1,2$, and the vector potential $A$ is small enough. Loss of uniqueness does occur if these assumptions are not satisfied, which is intimately linked to symmetry breaking phenomena~\cite{AftJerRoy-11,AshFroGraSchTro-02,CorPinRouYng-12,CorRinYng-07,GuoSei-13,Sei-02}. When the NLS functional satisfies some stability properties one may obtain error estimates on states, see Remark~\ref{rem:error states}. 

\medskip

The main virtue of the estimates of Theorems~\ref{thm:error Hartree} and ~\ref{thm:deriv nls} is that they do not depend on special properties of the interaction potentials (apart from the necessary stability conditions). They actually also do not depend on the vector potential $A$ and they would be exactly the same if the one-body term in the many-particle Hamiltonian is perturbed by any bounded operator. As we have said, there is however no reason to think that the rates of convergence are optimal.

When $d\geq 2$, the estimates of Theorem~\ref{thm:error Hartree} deteriorate with increasing $\beta$ and this results in a decrease of the range of applicability to the NLS limit. For the typical settings described above we obtain the desired result
\begin{equation}\label{eq:nls lim}
\frac{E(N)}{N} \to \eNLS \mbox{ when } N \to \infty
\end{equation}
under the conditions listed in Table~\ref{tab:NLS}. In 3D we distinguish the case where $w-\eta |w|$ is stable for some small fixed $\eta$ where we obtain better estimates than when $w$ is simply stable. In the former case we say that $w$ is $\eta$-stable. It remains an open problem to improve the convergence rates we obtain and to generalize the results to larger values of $\beta$. In particular, we expect that one should obtain the same NLS limit for all $0<\beta <1$ when $d\ge 2$, see~\cite{LieSeiYng-00,LieSeiYng-01} for the repulsive case. 

\begin{table}[htbp]
\begin{tabular}{|c|c|c||c|}
\hline
 & $d=3$, $w$ stable & $d=3$, $w$ $\eta$-stable & $d=2$ \\
\hline
$\phantom{\big\Vert}s=2$ & $\beta < 1/15$ & $\beta < 1/12$ & $\beta < 1/6 $ \\
\hline
$\phantom{\big|}s=\infty$ & $\beta < 2/21$  & $\beta < 2/15$ & $\beta < 1/4 $\\
\hline
\end{tabular}
\vspace{0.2cm}
\caption{Maximal value of $\beta$ in the NLS limit ($\beta>0$).}\label{tab:NLS} 
\end{table}

\vspace{-0.3cm}

The rest of the paper is organized as follows: the core of the analysis is in Section~\ref{sec:Hartree} where we prove Theorem~\ref{thm:error Hartree}. Next we turn to the proof of Proposition~\ref{thm:CV Hartree} in Section~\ref{sec:Hartree to NLS} where we also conclude the proofs of Theorems~\ref{thm:deriv nls} and~\ref{thm:nls state}. 

\section{Error bounds for Hartree theory}\label{sec:Hartree}

In this section, we prove Theorem~\ref{thm:error Hartree}.

\subsection{Quantum de Finetti and localization}

Let $\gamma_N$ be an arbitrary (mixed) state in the bosonic Hilbert space $\gH^N = \bigotimes_s ^N \gH$, i.e. a positive operator satisfying $\Tr \gamma_N=1$. For every $k=1,2,...,N$, the $k$-particle density matrix $\gamma_N^{(k)}$ is obtained by taking the partial trace over all but the first $k$ variables: 
$$
\gamma_N^{(k)}= \Tr_{k+1\to N} \left[\gamma_N\right]
$$
One of the main advantages of the reduced density matrices is that we can write 
\begin{align} \label{eq:EN/N-gamma2}
\frac{E(N)}{N}= \frac{1}{N}\left\langle \Psi_N, H_N \Psi_N \right \rangle =\frac{1}{2} \Tr_{\gH^2} \left[ H_2 \gamma_N^{(2)} \right]
\end{align}
where $\gamma_N=|\Psi_N\rangle\langle\Psi_N|$ and 
$$
H_{2}=- \left( \nabla_{x_1} +i A (x_1)\right) ^2 + V(x_1)- \left( \nabla_{x_2} +i A (x_2)\right) ^2 + V(x_2) + N^{d\beta}w(N^\beta(x_1-x_2)).
$$
When $\beta=0$, $H_2$ is independent of $N$ and the limit $E(N)/N\to\eH$
essentially comes from the structure of $\gamma_N^{(2)}$ in the large $N$ limit. If one does not need an error estimate, this convergence follows easily from a compactness argument and the quantum de Finetti theorem~\cite{Stormer-69,HudMoo-75}. This is  explained in~\cite[Section~3]{LewNamRou-14}. 

When $\beta>0$, we have to deal with a $N$-dependent interaction potential and the compactness argument in~\cite{LewNamRou-14} is not sufficient. Our strategy in this paper is to use a localization method to reduce the problem to a finite dimensional setting. We may then employ the following quantitative version of the quantum de Finetti theorem, originally proved in~\cite{ChrKonMitRen-07} (see~\cite{Chiribella-11,Harrow-13,LewNamRou-13b} for variants of the proof and~\cite{FanVan-06} for an earlier result in this direction):

\begin{theorem}[\textbf{Quantitative quantum de Finetti in finite dimension}]\label{thm:CKMR} \mbox{}\\
Let $\mathfrak{K}$ be a finite dimensional Hilbert space. For every state $G_N$ on $\mathfrak{K}^N:=\bigotimes_s^N\mathfrak{K}$ and for every $k=1,2,...,N$ we have
\begin{equation} \label{eq:error-CKMR-improved}
\Tr_{\mathfrak{K}^n} \Big| G_N ^{(k)} - \int_{S\mathfrak{K}} |u^{\otimes k}\rangle \langle u^{\otimes k}| d\mu_{G_N} (u)  \Big| \le 
\frac{4k \dim \mathfrak{K}}{N} 
\end{equation}
where $$G_N^{(n)}:=\Tr_{n+1\to N} \left[G_N\right]$$ and 
\begin{equation}\label{eq:deF mes CKMR}
d\mu_{G_N}(u) :=\dim \mathfrak{K}^N \pscal{u^{\otimes N},G_N u^{\otimes N}} du 
\end{equation}
with $du$ being the normalized uniform (Haar) measure on the unit sphere $S\mathfrak{K}$.
\end{theorem}

\begin{remark}\label{rem:Schur} The measure $\mu_{G_N}$ 
is a probability measure thanks to Schur's formula
$$ \1_{\mathfrak{K}^N} = \dim \mathfrak{K}^N \int_{S \mathfrak{K}}  |u^{\otimes k}\rangle \langle u^{\otimes k}| du.$$
We shall not need the explicit expression~\eqref{eq:deF mes CKMR}. If one could find a different construction giving a better error estimate, the convergence rates of our main theorems would be improved using the method we describe below. 
\hfill$\triangle$ 
\end{remark}

We will apply Theorem~\ref{thm:CKMR} to the low-lying energy subspace of the (magnetic) Schr\"odinger operator 
$H_1 = -\left( \nabla + i A \right)  ^2 + V$
acting on $\gH = L ^2 (\R ^d)$. We denote by $P_-$ and $P_+$ the spectral projectors above and below the energy cut-off $L$:
\begin{equation}\label{eq:projectors}
P_- = \1_{(-\infty, L)} \left( H_1 \right), \: P_+ = \1_{\gH} - P_-=P_- ^{\perp}. 
\end{equation}
Thanks to our assumption~\eqref{eq:asum V} and~\eqref{eq:asum A}, the dimension of the low-lying subspace
\begin{equation}\label{eq:NT}
N_L := \dim (P_- \gH ) = \mbox{ number of eigenvalues of } H_1 \mbox{ below } L
\end{equation}
is finite. Moreover it is controlled by a semi-classical inequality ``\`a la Cwikel-Lieb-Rosenblum'', stated in the next lemma. We refer to ~\cite[Chapter 4]{LieSei-09} for a thorough discussion of related inequalities.

\begin{lemma}[\textbf{Low-lying bound states of the one-body Hamiltonian}] \label{lem:nb bound states}\mbox{}\\
Let $V$ and $A$ satisfy~\eqref{eq:asum V} and~\eqref{eq:asum A}, respectively. Then for $L$ large enough we have
\begin{equation}\label{eq:bound NT}
N_L \leq C L ^{d/s + d/2}. 
\end{equation}
\end{lemma}

\begin{proof}
The number of eigenvalues of $(-i\nabla+A)^2+V$ below $L$ can be estimated by
$$N_L\leq \tr_{L^2(\R^d)}\left[ \exp \left(-\frac{(-i\nabla+A)^2+V-L}{L}\right)\right]\leq \frac{1}{(2\pi)^d}\iint_{\R^d\times\R^d}e^{-\frac{|p|^2+V-L}{L}}dx\,dp,$$
using~\cite[Thm. 2.1]{ComSchSei-78} and~\cite[Thm 15.8]{Simon-05}. Using our assumption~\eqref{eq:asum V} that $V(x)\geq c|x|^s-C$ and changing variables gives the result.
\end{proof}

We will combine the de Finetti theorem and the localization method in Fock space,  which provides the correct way of restricting a quantum $N$-body state to a subspace of $\gH$. Let us quickly recall this procedure, following the notation of~\cite{Lewin-11}.

Let $\gamma_{N}$  
be an arbitrary $N$-body (mixed) state.
With the given projections $P_-,P_+$, there are localized states $G_{N} ^-, G_N ^+ $ in the Fock space 
$$\cF(\gH)=\C\oplus\gH\oplus\gH^2\oplus\cdots$$
of the form
\begin{equation}
G_N ^{\pm} = G_{N,0}^ {\pm} \oplus G_{N,1}^ {\pm} \oplus\cdots\oplus G_{N,N}^ {\pm} \oplus0\oplus\cdots 
\label{eq:def_localization}
\end{equation}
with the crucial property that their reduced density matrices satisfy 
\begin{equation}
P_{\pm}^{\otimes n} \gamma^{(n)}_{N} P_{\pm} ^{\otimes n} = \left(G_N ^{\pm}\right)  ^{(n)}={N\choose n}^{-1}\sum_{k=n}^N{k\choose n}\tr_{n+1\to k}\left[G^{\pm}_{N,k}\right]
\label{eq:localized-DM} 
\end{equation}
for any $0 \leq n \leq N$. Here we use the convention that 
$$\gamma_N^{(n)}:=\Tr_{n+1\to N} [\gamma_N],$$
which differs from the convention of~\cite{Lewin-11}, whence the different numerical factors in~\eqref{eq:localized-DM}.

The relations \eqref{eq:localized-DM} determine the localized states $G_{N}^-, G_{N}^+$ uniquely and they ensure that $G_N ^-$ and $G_N ^+$ are (mixed) states on the Fock spaces $\cF (P_- \gH)$ and $\cF (P_+ \gH)$, respectively:
\bq\label{eq:nomalization-localized-state}
\sum_{k=0}^N \tr \left[ G_{N,k}^-\right] = \sum_{k=0}^N \tr \left[ G_{N,k}^+\right]=1.
\eq

Due to~\eqref{eq:EN/N-gamma2}, we are mainly interested in the two-particle density matrices. Applying the quantitative de Finetti Theorem~\ref{thm:CKMR} to the localized state $G_N^-$, we obtain the following

\begin{lemma} [\textbf{Quantitative quantum de Finetti for the localized state.}] \label{lem:deF-localized-state}\mbox{}\\
Let $\gamma_{N}$ be an arbitrary $N$-body (mixed) state and for every $L>0$, we have
$$
\Tr_{\gH^2} \left| P_-^{\otimes 2} \gamma_{N}^{(2)} P_-^{\otimes 2} - \int_{SP_-\gH} |u^{\otimes 2}\rangle \langle u^{\otimes 2}| d\mu_N(u)\right| \le \frac{8 N_L}{N}
$$
where
\bq \label{eq:def-mu-N-localized}
d\mu_N(u) = \sum_{k=2}^N\frac{k(k-1)}{N(N-1)}d\mu_{N,k}(u), \quad d\mu_{N,k}(u) =  \dim (P_-\gH)_s^k \pscal{u^{\otimes k},G_{N,k}^- u^{\otimes k}} du.
\eq
\end{lemma}

\begin{proof} Applying the quantitative de Finetti Theorem~\ref{thm:CKMR} we have 
\begin{align*}
\tr_{\gH ^2}\left| \Tr_{3\to k}\left[G_{N,k}  ^-\right]  - \int_{SP_- \gH} |u ^{\otimes 2}\rangle \langle u ^{\otimes 2} | d\mu_{N,k}(u) \right| \leq 8 \frac{N_L}{k} \tr_{\gH ^k} \left[ G_{N,k} ^-\right]
\end{align*}
where $d\mu_{N,k}(u)= \dim (P_-\gH)_s^k \langle u^{\otimes k},G_{N,k}^- u^{\otimes k}\rangle du$.
Combining this and~\eqref{eq:localized-DM} we get
\begin{align*}
\Tr_{\gH^2} \left| P_-^{\otimes 2} \gamma_{N}^{(2)} P_-^{\otimes 2} - \int_{SP_-\gH} |u^{\otimes 2}\rangle \langle u^{\otimes 2}| d\mu_N(u)\right|  
&\leq \sum_{k=2}^N {N \choose 2}^{-1} {k \choose 2} 
\frac{8N_L}{k} \tr_{\gH ^k} \left[ G_{N,k} ^-\right]\\
&= \frac{8N_L}{N}\sum_{k=2}^N  \frac{k-1}{N-1}\tr_{\gH ^k} \left[G_{N,k} ^-\right]\leq \frac{8N_L}{N},
\end{align*}
since $\sum_{k=0}^N \tr_{\gH ^k}[G_{N,k} ^-]=1$.
\end{proof}

\begin{remark}[De Finetti measure for the $n$-body reduced density matrix]\label{rem:construct deF}\mbox{}\\
We will later apply the same idea to any $n$-body reduced density matrix of $\gamma_N$ with $n\ge 2$. The de Finetti measures obtained in this way, namely
$$
d\mu_N ^n (u) = \sum_{k=n}^N  { N \choose n} ^{-1} {k\choose n}  d\mu_{N,k}(u)= \sum_{k = n}^N\frac{k(k-1)...(k-n+1)}{N(N-1)...(N-n+1)}d\mu_{N,k}(u),
$$
a priori depend on $n$, unless $G_{N,k} ^- = 0$ for any $k=0\ldots N-1$, namely unless all the particles are $P_-$ localized. However, in some situations (e.g. in our case in Section~\ref{sec:conv state}), we can show that most particles are $P_-$ localized, and thus all these $n$-dependent measures actually converge to the same limit.\hfill$\triangle$
\end{remark}

We shall prove Theorem~\ref{thm:error Hartree} in the following subsections. The core of the proof is in Section~\ref{sec:main terms}, where Lemma~\ref{lem:deF-localized-state} plays an essential role, with some preparation in Section~\ref{sec:loc} and some final computations in Section~\ref{sec:end proof}.

\subsection{Truncated two-body Hamiltonian}\label{sec:loc}

Recall from~\eqref{eq:EN/N-gamma2} that 
$$\frac{E(N)}{N}=\frac{1}{2}\Tr[ H_2 \gamma_N^{(2)}] $$
with $\gamma_N=|\Psi_N\rangle\langle\Psi_N|$. In order to reduce to a finite dimensional setting, we will replace the Hamiltonian $H_2$ by the localized operator $P_-\otimes P_- H_2 P_-\otimes P_-$, up to a small modification of the interaction potential $w_N$. The crucial fact is that if the energy cut-off $L$ is large enough (in comparison with $\|w_N\|_{L^\infty}$), then the kinetic energy of the $P_+$ localized particles can be used to control both the localization error and the number of $P_+$ localized particles.

For convenience we introduce a two-particle operator with modified interaction:
\begin{multline}
H_2^\epsilon=-(\nabla_{x_1}-iA(x_1))^2+V(x_1)-(\nabla_{x_2}-iA(x_2))^2+V(x_2)\\+w_N(x_1-x_2)-\epsilon|w_N(x_1-x_2)|.
\label{eq:def_H_Nepsilon}
\end{multline}
An important tool is then the following. 

\begin{lemma}[\textbf{Truncated two-body Hamiltonian}]\label{lem:localize-energy}\mbox{}\\
Assuming that $0<\epsilon\leq 1$ and $L \ge C N^{d\beta}\epsilon^{-1}$ (resp. $L\geq C\epsilon^{-2}$ if $d=1$ and $\beta>0$) for a large enough constant $C$, we have  
\begin{equation} \label{eq:H2-localized-error}
H_2 \ge  P_-^{\otimes 2}H_2^\epsilon P_-^{\otimes 2} + \frac{L}{2} (P_+H_1P_+\otimes \1 +\1\otimes P_+H_1P_+).
\end{equation}
\end{lemma}

The estimates in the $d=1$ case use a simple lemma that we prove at the end of this subsection:

\begin{lemma}[\textbf{Sobolev-type inequality in dimension $d=1$}]\label{lem:1D}\mbox{}\\
For any even potential $W\in L^1(\R)$, we have in dimension $d=1$
\begin{equation}\label{eq:1D op bound}
H_1 \otimes \1 + \1 \otimes H_1+W(x-y)\geq -C\left(\int_\R W^-\right)^2-C.
\end{equation}
\end{lemma}

\begin{proof}[Proof of Lemma~\ref{lem:localize-energy}] 
For the two-body non-interacting Hamiltonian 
$$H_1 \otimes \1 + \1 \otimes H_1=-(\nabla_{x_1}-iA(x_1))^2+V(x_1)-(\nabla_{x_2}-iA(x_2))^2+V(x_2),$$
we use that 
$$H_1=P_-H_1P_-+P_+H_1P_+$$
and obtain
\begin{align*}
 H_1 \otimes \1 + \1 \otimes H_1 &=  P_-H_1P_- \otimes \1 + \1 \otimes P_-H_1P_-  +  P_+H_1P_+ \otimes \1 + \1 \otimes P_+H_1P_+\\
&=(P_-)^{\otimes 2} \Big( H_1 \otimes \1 + \1 \otimes H_1 \Big) (P_-)^{\otimes 2}+P_-H_1P_- \otimes P_+ \\
& \quad + P_+ \otimes P_-H_1P_-  +  P_+H_1P_+ \otimes \1 + \1 \otimes P_+H_1P_+.
\end{align*}
Since $H_1$ is bounded from below, we have $P_-H_1P_-\geq -CP_-\geq -C$ and we obtain
\begin{equation}\label{eq:loc 1 body lower bound}
H_1 \otimes \1 + \1 \otimes H_1  \ge (P_-)^{\otimes 2} \Big( H_1 \otimes \1 + \1 \otimes H_1 \Big) (P_-)^{\otimes 2}  +  P_+(H_1-C)P_+ \otimes \1 + \1 \otimes P_+(H_1-C)P_+.
\end{equation}

We now consider the interaction term and use the shorthand notation $w_N$ for the two-body multiplication operator $(x_1,x_2)\mapsto w_N(x_1-x_2)$. We write again
\begin{align}
w_N &= \left( P_- + P_+ \right) ^{\otimes 2} w_N \left( P_- + P_+ \right) ^{\otimes 2}\nn\\
&= (P_-)^{\otimes 2}w_N(P_-)^{\otimes 2}+(P_-)^{\otimes 2}w_N\Pi+\Pi w_N(P_-)^{\otimes 2}+\Pi w_N\Pi,\label{eq:split interaction}
\end{align}
with the orthogonal projection $\Pi:=P_-\otimes P_++P_+\otimes P_-+P_+\otimes P_+$. To bound the  error terms we use the inequality
$$P A Q+ Q A P \geq -\epsilon P|A|P-\epsilon^{-1} Q|A|Q,$$ 
valid for any self-adjoint operator $A$, and any orthogonal projectors $P,Q$. If $A$ is positive this follows by writing\footnote{We copy the proof that the diagonal part of a positive hermitian matrix controls the off-diagonal part.}
$$ \left( \eps ^{1/2} P \pm \eps ^{-1/2} Q \right) A \left( \eps ^{1/2} P \pm \eps ^{-1/2} Q \right) \geq 0$$
and the general case is obtained by using the same bound applied to $A^+$ and $A ^-$ separately. 
We thereby deduce that
$$(P_-)^{\otimes 2}w_N\Pi+\Pi w_N(P_-)^{\otimes 2}\geq -\epsilon (P_-)^{\otimes 2}|w_N|(P_-)^{\otimes 2}-\epsilon^{-1}\Pi|w_N|\Pi$$
and, therefore, 
\begin{equation}
w_N\geq (P_-)^{\otimes 2}(w_N-\epsilon |w_N|)(P_-)^{\otimes 2}-(1+\epsilon^{-1})\Pi |w_N| \Pi.
\label{eq:proof d arbitrary}
\end{equation}
Using now $|w_N|\leq CN^{d\beta}$ and collecting our estimates we find 
\begin{align*}
H_2&\geq (P_-)^{\otimes 2}H^\epsilon_2(P_-)^{\otimes 2}+ P_+(H_1-C)P_+ \otimes \1 + \1 \otimes P_+(H_1-C)P_+-C(1+\epsilon^{-1})N^{d\beta}\Pi\\
&\geq (P_-)^{\otimes 2}H^\epsilon_2(P_-)^{\otimes 2}+\frac{1}{2}(P_+H_1P_+\otimes\1+\1\otimes P_+H_1P_+)+ \left(\frac{L}{2}-C-C(1+\epsilon^{-1}\right)N^{d\beta})\Pi
\end{align*}
and the result follows in dimensions $d\geq2$ or when $\beta=0$ and $d=1$.

When $d=1$ and $\beta>0$, we may come back to~\eqref{eq:proof d arbitrary} and use Lemma~\ref{lem:1D} to obtain
\begin{align*}
H_2&\geq (P_-)^{\otimes 2}H^\epsilon_2(P_-)^{\otimes 2}+\frac12(P_+H_1P_+\otimes\1+\1\otimes P_+H_1P_+)\\
&\qquad\qquad + \Pi\Big(L/4-C+(H_1\otimes\1+\1\otimes H_1)/4-(1+\epsilon^{-1}\big)|w_N|\Big)\Pi\\
&\geq (P_-)^{\otimes 2}H^\epsilon_2(P_-)^{\otimes 2}+\frac12(P_+H_1P_+\otimes\1+\1\otimes P_+H_1P_+)+ \Pi\big(L/4-C- C\epsilon^{-2}\big)\Pi.
\end{align*}
\end{proof}

\begin{proof}[Proof of Lemma~\ref{lem:1D}]
In 1D we can bound 
\begin{equation} \label{eq:rel-bound-Sobolev-1D}
\int_{\R}|W(x)| |u(x)|^2 \le \left( \int_{\R} |W| \right) \|u\|_{L^\infty(\R)}^2 \le C \|W\|_{L^{1}(\R)} \|u\|_{L^2(\R)} \left\| u' \right\|_{L^2(\R)}
\end{equation}
for every $u\in H^1(\R)$. We conclude that 
$$-d^2/dx^2+W\geq -d^2/dx^2-W^-\geq -C\norm{W^-}_{L^1(\R)}^2-C.$$
Using the pointwise diamagnetic inequality $|(\nabla + i A) u| \geq  | \nabla |u|\,|$ and removing the center of mass, the estimate is similar for the two-particle operator.
\end{proof}

\subsection{Bound on the localized energy}\label{sec:main terms} Now we turn to the main step of the proof of Theorem~\ref{thm:error Hartree}: we compare the energy corresponding to the localized operator in the right side of (\ref{eq:H2-localized-error}) with the Hartree energy $\eH$. 

Let $\gamma_{N}=|\Psi_N \rangle \langle \Psi_N|$ be the density matrix of a ground state $\Psi_N \in \gH^N$ of the $N$-body Hamiltonian $H _N$ (known to exist in our setting of trapped systems). We have 
\begin{lemma}[\textbf{Lower bound to the localized energy}]\label{lem:main terms}\mbox{}\\
For $0<\epsilon\leq 1$, $L \ge C N^{d\beta}\epsilon^{-1}$ (resp. $L\geq C\epsilon^{-2}$ if $d=1$ and $\beta>0$) and $N$ large enough, we have  
\begin{equation}\label{eq:low bound main}
\frac{1}{2}\tr \left[ P_-^{\otimes 2} H_2^\epsilon P_-^{\otimes 2}\gamma_{N}^{(2)} \right]+\frac{L}4\tr \left[ P_+\gamma_N^{(1)} \right] \ge \eH^\epsilon - \frac{CL^{1+d/s+d/2}}{N}.
\end{equation}
\end{lemma}

\begin{proof} 
We will employ the same notations for localized states $G_N^\pm$ obtained from the $N$-body state $\gamma_N$ and the projections $P_\pm$ as in Subsection 3.1. By Lemma~\ref{lem:deF-localized-state}, we have 
$$
\Tr_{\gH^2} \left| P_-^{\otimes 2} \gamma_{N}^{(2)} P_-^{\otimes 2} - \int_{SP_-\gH} |u^{\otimes 2}\rangle \langle u^{\otimes 2}| d\mu_N(u)\right| \le \frac{8 N_L}{N}
$$
with $d\mu_N$ defined as in~\eqref{eq:def-mu-N-localized}. On the other hand, 
\begin{equation}\label{eq:bound H-}
\norm{P_-^{\otimes 2} H_2^\epsilon P_-^{\otimes 2} } \leq 2L + (1+\epsilon)\|w_N\|_{L^\infty} \le CL 
\end{equation}
in operator norm because we have truncated the high energy spectrum of the one-body part and the two-body potential is bounded by $\|w_N\|_{L^\infty}\le CN^{d\beta}\le CL$. 
In dimension $d=1$, the estimate is the same but we use that $|w_N|\leq H_1 \otimes \1 + \1 \otimes H_1+C$, by Lemma~\ref{lem:1D}.
Therefore, 
\begin{align} 
\frac{1}{2}\tr \left[ P_-^{\otimes 2} H_2^\epsilon P_-^{\otimes 2}\gamma_{N}^{(2)} \right] &\ge \frac{1}{2}\int_{SP_- \gH} \tr_{\gH ^2} \left[ H^\epsilon_2 |u ^{\otimes 2}\rangle \langle u ^{\otimes 2}| \right]  d\mu_N -\frac{CLN_L}{N} \nn\\
& \geq \int_{SP_- \gH} \EH^\epsilon [u] d\mu_N -\frac{CL^{1+d/s+d/2}}{N}\nn\\
& \geq \int_{SP_- \gH} \big(\EH^\epsilon [u]-\eH^\epsilon\big) d\mu_N +\eH ^\epsilon \,\mu_N(SP_-\gH)-\frac{CL^{1+d/s+d/2}}{N}\label{eq:neglect integral}
\end{align}
where we have used the estimate $N_L\leq CL^{d/s+d/2}$ of Lemma~\ref{lem:nb bound states}.
Applying the variational principle $\EH^\epsilon [u] \geq \eH^\epsilon$ and computing $\int d\mu_N$ explicitly using~\eqref{eq:def-mu-N-localized}, we get
\begin{align} \label{eq:P-term}
\frac{1}{2}\tr \left[ P_-^{\otimes 2} H^\epsilon_2 P_-^{\otimes 2} \gamma_{N}^{(2)} \right] \ge 
\eH^\epsilon \sum_{k=2} ^N \frac{k(k-1)}{N(N-1)}\tr_{\gH ^k} \left[ G_{N,k} ^- \right] -\frac{CL^{1+d/s+d/2}}{N}.
\end{align}

If $\eH^\epsilon\leq0$ then we simply write
$$\sum_{k=2} ^N \frac{k(k-1)}{N(N-1)}\tr_{\gH ^k} \left[ G_{N,k} ^- \right]\leq \sum_{k=2} ^N \tr_{\gH ^k} \left[ G_{N,k} ^- \right]\leq1$$
and we are done. If $\eH^\epsilon>0$, then we need to prove that $\int d\mu_N\simeq1$ and for this we use the positive term 
$$\tr \left[(P_+\otimes\1+\1\otimes P_+) \gamma_N^{(2)}\right]=2\tr \left[P_+\gamma_N^{(1)} \right]=2\left(1-\tr\left[P_-\gamma_N^{(1)}\right]\right).$$
First, recall that by~\eqref{eq:localized-DM}
$$ \tr \left[P_- \gamma_N^{(1)} \right] = \sum_{k=1} ^N \frac{k}{N} \tr_{\gH ^k} \left[G_{N,k} ^-\right].$$
Then, using Jensen's inequality, we have
\begin{align*}
\sum_{k=2} ^N \frac{k(k-1)}{N(N-1)}\tr_{\gH ^k} \left[ G_{N,k} ^- \right]&= \frac{N}{N-1}\sum_{k=0} ^N \frac{k^2}{N^2}\tr_{\gH ^k} \left[ G_{N,k} ^- \right]-\frac{\tr P_-\gamma_N^{(1)}}{N-1}\\
&\geq\frac{N}{N-1}\left(\sum_{k=0} ^N \frac{k}{N}\tr_{\gH ^k} \left[ G_{N,k} ^- \right]\right)^2-\frac{\tr P_-\gamma_N^{(1)}}{N-1}\\
&=\frac{N}{N-1}\left(\tr P_-\gamma_N^{(1)}\right)^2-\frac{\tr P_-\gamma_N^{(1)}}{N-1}.
\end{align*}
Therefore, denoting by 
$$\lambda:=\tr P_-\gamma_N^{(1)}\leq1$$
the fraction of particles localized on the low energy states, we have the lower bound
$$\frac{1}{2}\tr \left[ P_-^{\otimes 2} H_2^\epsilon P_-^{\otimes 2}\gamma_{N}^{(2)} \right]+\frac{L}4\tr P_+\gamma_N^{(1)}\geq \frac{N\lambda^2-\lambda}{N-1}\eH^\epsilon+\frac{L}4(1-\lambda)-\frac{CL^{1+d/s+d/2}}{N}.$$
When $L>8N\eH^\epsilon/(N-1)$, the minimum of the right side is attained at $\lambda=1$. In dimension $d\geq2$, we have $\eH^\epsilon\leq CN^{d\beta}$ and $L\geq CN^{d\beta}/\epsilon$, and hence the condition $L>8N\eH^\epsilon/(N-1)$ is always fulfilled for $N$ large enough. Similarly, if $d=1$ and $\beta>0$, we use that $|\eH^\epsilon|\leq C$ and $L\geq C\epsilon^{-2}$. In all cases we get 
\begin{equation}
\frac{1}{2}\tr \left[ P_-^{\otimes 2} H_2^\epsilon P_-^{\otimes 2}\gamma_{N}^{(2)} \right]+\frac{L}4\tr \left[ P_+\gamma_N^{(1)} \right] \geq \eH^\epsilon-\frac{CL^{1+d/s+d/2}}{N}
\label{eq:final estim with control on excited particles}
\end{equation}
which is the desired inequality.
\end{proof}

\subsection{Conclusion: proof of Theorem~\ref{thm:error Hartree}}\label{sec:end proof}

The upper bound in~\eqref{eq:energy estimate} is trivial, taking a factorized trial state for the $N$-body energy. To show the lower bound, let us consider a ground state $\gamma_{N}=|\Psi_N \rangle \langle \Psi_N|$ of $H_N$. From Lemma~\ref{lem:localize-energy} and Lemma~\ref{lem:main terms}, we have the lower bound 
\begin{equation}\label{eq:low additional term}
\frac{E(N)}{N} \geq \eH^\epsilon -C\frac{L^{1+ d/s + d/2}}{N}+\frac{1}4\tr \left[P_+H_1\gamma_N^{(1)}\right]
\end{equation}
which gives the result for $L = CN^{d\beta}\epsilon^{-1}$  in dimensions $d\geq2$ and, if $\beta=0$, in dimension $d=1$. The last term in~\eqref{eq:low additional term} is positive and will be useful later. For now we can just drop it from the lower bound.

If $d=1$ and $\beta>0$, we take $L = C\epsilon^{-2}$. By Lemma~\ref{lem:1D} we have $|\eH|\leq C$ and 
$$H_2^\epsilon=(1-\epsilon)H_2+\epsilon\left(H_1 \otimes \1 + \1 \otimes H_1-(w_N)_-\right)\geq (1-\epsilon)H_2-\epsilon C,$$
which implies $\eH^\epsilon\geq \eH-C\epsilon$. Thus we find
$$\frac{E(N)}{N} \geq \eH-C\epsilon-C\frac{\epsilon^{-3-2/s}}{N}+\frac{1}4\tr \left[P_+H_1\gamma_N^{(1)}\right].$$
The estimate~\eqref{eq:energy estimate 1D} then follows after optimizing with respect to $\epsilon$.\hfill\qed

\section{The NLS limit}\label{sec:Hartree to NLS}

Now we explain how to go from the Hartree model to the NLS model and prove Proposition~\ref{thm:CV Hartree} and Theorems~\ref{thm:deriv nls} and~\ref{thm:nls state}.

\subsection{From Hartree to NLS: proof of Proposition~\ref{thm:CV Hartree}} The key observations are summarized in the following lemma.

\begin{lemma}[\textbf{Stability of the effective one-body functionals}]\label{lem:Hartree NLS}\mbox{}\\
Let $d\leq 3$. We assume that~\eqref{eq:asum w},~\eqref{eq:asum V} and~\eqref{eq:asum A} hold true, and that $|x| w (x) \in L ^1 (\R ^d)$. When $d=2,3$ we also assume that $w$ is stable (in the sense of~\eqref{eq:Hartree stable strict} when $d=2$ and~\eqref{eq:repulsive} when $d=3$).

Then the set of minimizers for $\ENLS$ is non-empty and compact in the quadratic form domain of $H_1$. Also, for any normalized function $u\in L^2(\R^d)$ we have
\bq \label{eq:hartree NLS 0}
\int_{\R^d}\left|\nabla|u|(x)\right|^2\,dx\le C (\EH [u]+C)
\eq
and 
\bq \label{eq:hartree NLS 1}
\left| \EH [u] - \ENLS[u] \right| \le C N^{-\beta} \left(1+\int_{\R^d}\left|\nabla|u|(x)\right|^2\,dx\right)^2. 
\eq
\end{lemma}


\begin{proof}
We start by proving~\eqref{eq:hartree NLS 0}. When $d=3$, since $w$ satisfies~\eqref{eq:repulsive} the nonlinear term is positive, and therefore the inequality~\eqref{eq:hartree NLS 0} follows immediately from the diamagnetic inequality $|(\nabla + i A) u| \geq  | \nabla |u|\,|$ and the assumption that $V\geq -C$.

If $d=2$, we use the stability assumption on $w$ to obtain, for $\eta$ small enough,
\begin{align*}
C+\EH(u)&\geq \int_{\R^2}|\nabla|u||^2+\frac12\iint_{\R^2\times\R^2}|u(x)|^2|u(y)|^2w(x-y)\,dx\,dy\\
&\geq \frac{1}{1- \eta}\left(\int_{\R^2}|\nabla|u||^2+\frac{1-\eta}2\iint_{\R^2\times\R^2} |u(x)|^2|u(y)|^2w(x-y)\,dx\,dy\right)\\
&\qquad\qquad+\frac{\eta}{1- \eta}\int_{\R^2}|\nabla|u||^2\\
&\geq \frac{\eta}{1- \eta}\int_{\R^2}|\nabla|u||^2.
\end{align*}
Finally, in dimension $d=1$, we use Lemma~\ref{lem:1D} and obtain
$$
\EH[u]=\frac12\pscal{u^{\otimes 2},H_2u^{\otimes 2}}
\geq -C-C\left(\int_\R w^-\right)^2+\frac12\pscal{u,H_1u}\geq \frac12\int_{\R}|u'|^2-C.
$$

Next, we prove (\ref{eq:hartree NLS 1}). We change variables to write 
\begin{align*}
\EH [u]- \ENLS [u] &= \frac{1}{2}\iint_{\R ^d\times \R ^d} |u(x)| ^2 w_N (x-y) |u(y)| ^2 dxdy - \frac{1}{2} \left( \int w \right)\int_{\R ^d } |u|^4 \\
& = \frac{1}{2} \iint_{\R ^d\times \R ^d} |u(x)|^2  w(z) \Big( |u(x+zN^{-\beta})|^2-|u(x)|^2 \Big) dxdz\\
&= \frac{1}{2} \iint_{\R ^d\times \R ^d}  |u(x)|^2  w(z) \Big( \int_{0}^1  \nabla |u|^2 (x+tzN^{-\beta}) \cdot z N^{-\beta} dt \Big) dx dz
\end{align*}
and hence
\begin{equation*}
\big|\EH [u]- \ENLS [u]\big|\leq N^{-\beta}\left(\int_{\R^d}|z||w(z)|\,dz\right)\norm{|u|^2\ast\nabla|u|^2}_{L^\ii(\R^d)}.
\end{equation*}
In dimensions $d\leq3$, we can write by the Young and Sobolev inequalities
\begin{align*}
\norm{|u|^2\ast|\nabla|u|^2|}_{L^\ii(\R^d)}&\leq \norm{u}^3_{L^{6}}\norm{\nabla |u|}_{L^{2}(\R^d)}\leq C\norm{|u|}_{H^1(\R^d)}^4,
\end{align*}
which concludes the proof of~\eqref{eq:hartree NLS 1}.
\end{proof}

With Lemma~\ref{lem:Hartree NLS} at hand, it is now easy to prove Proposition~\ref{thm:CV Hartree}.

\begin{proof}[Proof of Proposition~\ref{thm:CV Hartree}]
We first prove 
\begin{equation}\label{eq:hartree NLS 2}
|\eH - \eNLS | \le C N ^{-\beta}.
\end{equation}
Let $v$ be a minimizer for $\eNLS$. Then $|v|\in H^1(\R^d)$ by~\eqref{eq:hartree NLS 0}. Consequently, 
$$\eH\leq \EH(v)\leq \ENLS(v)+CN^{-\beta}\norm{|v|}_{H^1(\R^d)}^4=\eNLS+CN^{-\beta}.$$
Similarly, pick for every $N$ a minimizer $u_N$ for $\eH$. By~\eqref{eq:hartree NLS 0}, $|u_N|$ is uniformly bounded in $H^1(\R^d)$. Hence
$$\eNLS\leq \ENLS(u_N)\leq \EH(u_N)+CN^{-\beta}\norm{|u_N|}_{H^1(\R^d)}^4\leq \eH+CN^{-\beta},$$
which concludes the proof of~\eqref{eq:hartree NLS 2}.

Next we prove that $\eH\to-\ii$ as $N\to\ii$ when $w$ is not stable in dimensions $d\geq2$. Let $u\in C^\ii_c(\R^d)$ with support in the unit ball and $\norm{u}_{L^2}=1$. Using as trial state $v_N(x)=N^{d\beta/2} u(N^\beta x)$, our assumption~\eqref{eq:asum V} on $V$ and the pointwise estimate 
$$|\nabla v_N+iAv_N|^2\leq (1+\epsilon)|\nabla v_N|^2+(1+\epsilon^{-1})|A|^2|v_N|^2,$$ 
we find that
\begin{align*}
\eH
&\leq C+(1+\epsilon)N^{2\beta}\int_{\R^d} |\nabla u|^2+(1+\epsilon^{-1})N^{d\beta}\norm{u}_{L^\ii}\int_{|x|N^\beta\leq 1} |A|^2 + CN^{-S\beta}\int_{\R^d}|x|^S|u|^2\\
&\qquad\qquad\qquad+\frac{N^{d\beta}}{2} \iint _{\R^d \times \R^d} |u(x)|^2 w(x-y) |u(y)|^2 dx dy \\
&=(1+\epsilon)N^{2\beta}\int_{\R^d} |\nabla u|^2+\frac{N^{d\beta}}{2} \iint _{\R^d \times \R^d} |u(x)|^2 w(x-y) |u(y)|^2 dx dy+o(N^{d\beta}).
\end{align*}
If $d\geq3$ and $w$ is not classically stable, we choose $u$ to have
$$
\iint _{\R^d \times \R^d} |u(x)|^2 w(x-y) |u(y)|^2 dx dy  <0,
$$
and we conclude that $\lim_{N\to\ii}\eH=\lim_{N\to\ii}E(N)/N=-\ii$ since $E(N)/N\leq\eH$. If $d=2$, we take $\epsilon^2=\int_{|x|\leq N^{-\beta}}|A|^2$ and obtain 
$$\limsup_{N\to\ii}\frac{\eH}{N^{2\beta}}\leq \inf_{\substack{u\in H^1(\R^2)\\ \norm{u}_{L^2(\R^2)}=1}}\left(\int_{\R^2} |\nabla u|^2+\frac{1}{2} \iint _{\R^2 \times \R^2} |u(x)|^2 w(x-y) |u(y)|^2 dx dy\right).$$
The right side is strictly negative by assumption.
\end{proof}

\subsection{Convergence of the many-body energy: proof of Theorem~\ref{thm:deriv nls}}
In the cases $d=1$, $d=2$ with $w$ stable, the result immediately follows from Theorem~\ref{thm:error Hartree}, Remark~\ref{rem:MF estimates} and Proposition~\ref{thm:CV Hartree}. If $w-\eta|w|$ is stable for some $0<\eta<1$, so is $w-\epsilon|w|$ for all $0 \leq \epsilon \leq \eta$. Then we can apply Proposition~\ref{thm:CV Hartree} with $w$ replaced by $w-\epsilon|w|$ for some $\epsilon<1$ to be tuned later on and get
\begin{equation}
|\eNLS^\epsilon-\eH^\epsilon|\leq CN^{-\beta}
\label{eq:compare NLS epsilon}
\end{equation}
where $\eNLS^\epsilon$ is the NLS minimization problem, with $a=\int w$ replaced by $\int w-\epsilon\int|w|$. Arguing as in the proof of Lemma~\ref{lem:Hartree NLS}, it is not difficult to see that $|\eNLS^\epsilon-\eNLS|\leq C\epsilon$. Therefore Theorem~\ref{thm:error Hartree} provides the bound
$$\eNLS+CN^{-\beta}\geq \frac{E(N)}{N}\geq \eNLS-C\epsilon-CN^{-\beta}-C\frac{\epsilon^{-1-d/2-d/s}}{N ^{1-d\beta(1+d/2+d/s)}}.$$
Optimizing over $\epsilon$ gives~\eqref{eq:energy estimate NLS}.
\hfill\qed

\subsection{Convergence of states: proof of Theorem~\ref{thm:nls state}}\label{sec:conv state} We split the proof in four steps for clarity.

\noindent\textbf{Step 1, strong compactness of density matrices.} We first note that $\gamma_N ^{(n)}$ is by definition bounded in the trace-class, so that we can extract a subsequence along which
\begin{equation}\label{eq:weak CV}
\gamma_N ^{(n)} \wto_* \gamma ^{(n)} 
\end{equation}
as $N\to \infty$. Modulo a diagonal extraction argument, one can assume that the convergence is along the same subsequence for any $n$. We now argue that the convergence is actually strong. We start by proving that 
\begin{equation}\label{eq:unif bound kinetic}
\tr \Big[H_1\gamma_N^{(1)}\Big]=\Tr \Big[ \left(- \big( \nabla +i A \right) ^2 + V\big) \gamma_N^{(1)} \Big] \leq C, 
\end{equation}
independently of $N$. To this end, pick some $\alpha>0$ and define
\begin{equation*}\label{eq:start hamil eta}
H_{N,\alpha} = \sum_{j=1} ^N  \left( - \Big( \nabla +i A (x_j)\right) ^2 + V(x_j) \Big) + \frac{1+\alpha}{N-1} \sum_{1\leq i<j \leq N}N ^{d\beta} w( N ^{\beta} (x_i-x_j)).
\end{equation*}
Noticing for instance that
$$(1+\eta/4)(w-\eta/4|w|)\geq w-(\eta/2+\eta^2/16)|w|\geq w-\eta|w|,$$
we can apply the results of Theorem~\ref{thm:deriv nls} with $H_N$ replaced by $H_{N,\alpha}$. We find in particular that $H_{N,\alpha}\geq -CN$ and deduce that
$$\eNLS+o(1)\geq \frac{\pscal{\Psi_N,H_N\Psi_N}}N\geq -C(1+\alpha)^{-1}+\frac{\alpha}{1+\alpha}\tr \big[H_1\gamma_N^{(1)}\big].$$
Hence~\eqref{eq:unif bound kinetic} holds true.
Since $H_1=- \left( \nabla +i A \right) ^2 + V$ has a compact resolvent,~\eqref{eq:weak CV} and~\eqref{eq:unif bound kinetic} imply that, up to a subsequence, $\gamma_{N}^{(1)}$ converges strongly in the trace class. By~\cite[Corollary~2.4]{LewNamRou-14}, $\gamma_N^{(n)}$ converges strongly as well for all $n\geq1$.

\medskip

\noindent\textbf{Step 2, introducing the limit measure.} Now we extract some useful information from the proof of Theorem~\ref{thm:error Hartree}. We shall use the same notation as in Section ~\ref{sec:Hartree} (with the same choices for $L$ and $\epsilon$ that were made later in the proof). For simplicity we denote by
$$r_N:=N^{-\beta}+N^{-\tfrac{1}{4+2/s}}\1(d=1)+N^{-\tfrac{1-d\beta(1+d/2+d/s)}{2+d/s+d/2}}\1(d=2,3)$$
the best error bound that we have derived on $|E(N)/N-\eNLS|$.

Let $d\mu_N$ be defined as in Lemma~\ref{lem:deF-localized-state}, which is such that 
$$\mu_N(SP_-\gH)=\tr \left[P_-^{\otimes 2}\gamma_N^{(2)}P_-^{\otimes 2}\right]$$
Our arguments of Section~\ref{sec:Hartree} actually imply that $\mu_N(SP_-\gH)\to1$ but we will recover this fact here. We have 
$$
\Tr \left| P_-^{\otimes 2} \gamma_N^{(2)} P_-^{\otimes 2} - \int_{SP_-\gH} |u^{\otimes 2}\rangle \langle u^{\otimes 2}| d\mu_N (u) \right| \leq \frac{8N_L}{N}\leq C\frac{L^{1+d/s+d/2}}{N}\to0.
$$
On the other hand, the estimate~\eqref{eq:final estim with control on excited particles} and the error bounds provide a control on the number of excited particles: 
\begin{equation}\label{eq:bound excited particles}
1-\mu_N(SP_-\gH)=\Tr\left[(1-P_-^{\otimes 2}) \gamma_N^{(2)}\right] \leq 2\Tr \left[P_+ \gamma_N^{(1)} \right]\leq \frac{r_N}{L}. 
\end{equation}
Therefore, by the triangle and Cauchy-Schwarz inequalities, we find
\bq \label{eq:mu-N-localized-cv-gamma2}
\Tr \left|\gamma_N^{(2)} - \int_{SP_-\gH} |u^{\otimes 2}\rangle \langle u^{\otimes 2}| d\mu_N (u) \right| \leq C\frac{L^{1+d/s+d/2}}{N}+C\sqrt{\frac{r_N}{L}}.
\eq
Next, we denote $P_K$ the spectral projector of $H_1$ onto energies below a cut-off $K$. Since $\gamma_N ^{(2)} \to \gamma ^{(2)}$ and $P_K \to \1$ we deduce from the above 
$$ 
\lim_{K \to \infty} \lim_{N\to \infty} \mu_N (S P_K\gH) = 1.
$$
This tightness condition allows us to use Prokhorov's theorem and~\cite[Lemma~1]{Skorokhod-74} to ensure that, up to extraction of a subsequence, $\mu_N$ converges weakly to a measure $\mu$ in the ball $B\gH$. After passing to the weak limit, we find that
$$\gamma^{(2)}=\int_{B\gH} |u^{\otimes 2}\rangle \langle u^{\otimes 2}| d\mu(u).$$
Since $\mu(B\gH)\leq1$ and $\tr\gamma^{(2)}=1$ by the strong convergence of $\gamma^{(2)}_N$, we conclude that $\mu$ is supported on the sphere $S\gH$.

\medskip

\noindent\textbf{Step 3, the limit measure charges only NLS minimizers.} Going back to our proof in the previous sections and using that 
$$\mu_N(SP_-\gH)=1+O\left(\frac{r_N}{L}\right),$$
we deduce the bound 
$$\int_{SP_-\gH } \big(\EH^\epsilon [u] -\eH^\epsilon\big)d\mu_N (u)\leq Cr_N$$
on the term that we had neglected in~\eqref{eq:neglect integral}.
By Lemma~\ref{lem:Hartree NLS}, this implies that, for $B$ a large enough (but fixed) constant,
$$\frac{B^2}{C}\int_{\|\nabla|u|\|_{H^1}\geq B} d\mu_N (u)\leq \int_{\|\nabla|u|\|_{H^1}\geq B} \big(\EH^\epsilon [u] -\eH^\epsilon\big)d\mu_N (u)\leq Cr_N.$$
and
$$\int_{\|\nabla|u|\|_{L^2}\leq B} \big(\ENLS [u] -\eNLS\big)d\mu_N\leq C(1+B^4)(\epsilon+N^{-\beta})+\int_{\|\nabla|u|\|_{L^2}\leq B} \big(\EH [u] -\eH\big)d\mu_N (u)\leq Cr_N.$$
Passing to the limit $N\to\ii$, it is now clear that $\mu$ has its support on $\MNLS$.

At this stage, from~\eqref{eq:mu-N-localized-cv-gamma2} and the convergence of $\mu_N$ we have, along a subsequence  
$$ \gamma_N^{(2)} \to \int_{\MNLS} |u^{\otimes 2}\rangle \langle u^{\otimes 2}| d\mu(u),$$
strongly in the trace-class, where $\mu$ is a probability measure supported on $\MNLS$. Taking a partial trace we also have
$$ \gamma_N^{(1)} \to \int_{\MNLS} |u\rangle \langle u| d\mu(u).$$

\medskip

\noindent\textbf{Step 4, higher order density matrices.} There remains to prove that, for any $n > 2$,
$$ \gamma_N^{(n)} \to \int_{\MNLS} |u^{\otimes n}\rangle \langle u^{\otimes n}| d\mu(u),$$
strongly in the trace-class when $N\to \infty$. In view of the definition of $\mu$ this follows from the estimate
\begin{equation}\label{eq:claim higher DM}
\Tr \left|\gamma_N^{(n)} - \int_{SP_-\gH} |u^{\otimes n}\rangle \langle u^{\otimes n}| d\mu_N (u) \right| \to 0.
\end{equation}
To see that~\eqref{eq:claim higher DM} holds we first define a measure approximating $\gamma_N ^{(n)}$, as indicated in Remark~\ref{rem:construct deF}. Arguing as in the proof of Lemma~\ref{lem:deF-localized-state} we have  
\begin{equation}\label{eq:deF higher DM}
\Tr_{\gH^n} \left| P_-^{\otimes n} \gamma_{N}^{(n)} P_-^{\otimes n} - \int_{SP_-\gH} |u^{\otimes n}\rangle \langle u^{\otimes n}| d\mu_N ^n (u)\right| \le C\frac{ n N_L}{N}
\end{equation}
where
\bq \label{eq:def-mu-N-localized bis}
d\mu_N ^n (u) = \sum_{k=n}^N  { N \choose n} ^{-1} {k\choose n}  d\mu_{N,k}(u)
\eq
and $d\mu_{N,k}$ is the same measure as in~\eqref{eq:def-mu-N-localized}. An estimate similar to~\eqref{eq:bound excited particles} next shows that 
$$\Tr_{\gH^n} \left| \gamma_{N}^{(n)} - \int_{SP_-\gH} |u^{\otimes n}\rangle \langle u^{\otimes n}| d\mu_N ^n (u)\right| \to 0.$$
Using the easy bound (see~\cite[Section 2]{LewNamRou-14})
$$ { N \choose n} ^{-1} {k\choose n} = \left(\frac{k}{N}\right) ^n + O (N ^{-1})$$
together with the triangle inequality and Schur's formula (see Remark~\ref{rem:Schur}) we next deduce from~\eqref{eq:deF higher DM} that 
\begin{align}\label{eq:higher DM final}
\Tr \left|\gamma_N^{(n)} - \int_{SP_-\gH} |u^{\otimes n}\rangle \langle u^{\otimes n}| d\mu_N (u) \right| &\leq  \sum_{k=0} ^N \left( \left(\frac{k}{N}\right) ^2 - \left(\frac{k}{N}\right) ^n \right) \tr_{\gH ^k} \left[ G_{N,k} ^- \right]  \nonumber
\\&+ \sum_{k=0} ^{n-1} \left(\frac{k}{N}\right) ^n \tr_{\gH ^k} \left[ G_{N,k} ^- \right] \nonumber
\\&+ \sum_{k=0} ^{2} \left(\frac{k}{N}\right) ^2 \tr_{\gH ^k} \left[ G_{N,k} ^- \right] + o(1).
\end{align}
Finally, combining our previous estimates gives  
$$ \sum_{k=2} ^N  \left(\frac{k}{N}\right) ^2 \tr_{\gH ^k} \left[ G_{N,k} ^- \right] \to 1$$
and using~\eqref{eq:nomalization-localized-state} we in fact have
$$\sum_{k=0} ^N  \left(\frac{k}{N}\right) ^2 \tr_{\gH ^k} \left[ G_{N,k} ^- \right] \to 1.$$
Then by Jensen's inequality  
$$1\geq \sum_{k=0} ^N  \left(\frac{k}{N}\right) ^n \tr_{\gH ^k} \left[ G_{N,k} ^- \right]\geq \left( \sum_{k=0} ^N  \left(\frac{k}{N}\right) ^2 \tr_{\gH ^k} \left[ G_{N,k} ^- \right]\right) ^{n/2} \to 1.$$
Inserting this and~\eqref{eq:nomalization-localized-state} in~\eqref{eq:higher DM final} concludes the proof of~\eqref{eq:claim higher DM} and thus that of the theorem. 
\hfill\qed

We finally note that our method can give quantitative estimates on reduced density matrices in some special cases:

\begin{remark}[The case of a stable NLS functional]\label{rem:error states}
If a stability estimate of the form
\begin{equation}\label{eq:stability NLS}
\ENLS [u] \geq \eNLS + c \inf_{v\in \MNLS} \norm{u-v} ^2 
\end{equation}
holds at the level of the NLS functional, for some norm $\norm{\:.\:}$ (say the $L^2$ norm), it is easy to see that the previous method leads to quantitative estimates on density matrices. In particular, if the NLS minimizer is unique (up to a constant phase), non-degenerate, and~\eqref{eq:stability NLS} holds, one may obtain quantitative bounds on the depletion of the condensate. Indeed, it immediately follows from the above considerations and~\eqref{eq:stability NLS} that  
$$ \int_{u\in SP_- \gH} \big\Vert \left|u\right\rangle\left\langle u\right| - \left|\uNLS \right\rangle\left\langle \uNLS\right| \big\Vert_{\gS ^2} ^2 \:d\mu_N (u) \leq r_N$$
where $\gS ^2$ denotes the Hilbert-Schmidt class. Then, using Jensen's inequality together with the fact that $|\uNLS \rangle\langle \uNLS|$ is a rank-one projection we have 
$$ \left\Vert \int_{u\in SP_- \gH} |u\rangle \langle u| d\mu_N (u) - |\uNLS \rangle\langle \uNLS| \right\Vert_{\gS ^1} \leq C \sqrt{r_N}$$
and thus, in view of our previous bounds, 
$$ \left\Vert \gamma_N ^{(n)} - \left|\uNLS  ^{\otimes n}\right\rangle\left\langle \uNLS ^{\otimes n} \right| \right\Vert_{\gS ^1 (L ^2 (\R^3))} \leq C_n \sqrt{r_N},$$
where $C_n$ depends only on $n$.

Since we do not want to rely on assumptions such as~\eqref{eq:stability NLS} we do not pursue in this direction. One may consult e.g.~\cite{CarFraLie-14} or~\cite[Section 6]{Frank-14} for discussions of estimates of the form~\eqref{eq:stability NLS}. $\triangle$
\end{remark}


\end{document}